\documentclass[a4paper,11pt]{article}

\usepackage[utf8]{inputenc}
\usepackage[T1]{fontenc}
\usepackage[english]{babel}

\usepackage[usenames]{xcolor}

\usepackage{tikz}
\usetikzlibrary{calc}

\usepackage{amsmath,amssymb,amsfonts,amsthm}
\usepackage{mathabx}
\usepackage[ruled,linesnumbered]{algorithm2e}

\usepackage{cleveref}

\usepackage{fullpage}
\usepackage{enumitem}

\newcommand{\eps}{\varepsilon}
\newcommand{\Ot}{\tilde{O}}
\def\polyloglog{\operatorname{polyloglog}}

\newcommand{\Vor}[2]{\ensuremath{\mbox{Vor}({#1},{#2})}}
\newcommand{\Trivor}[2]{\ensuremath{{#1}_{\Delta}({#2})}}

\newcommand{\calV}{\mathcal{V}}
\newcommand{\calH}{\mathcal{H}}
\newcommand{\calR}{\mathcal{R}}
\newcommand{\dist}{\text{dist}}

\newtheorem{lemma}{Lemma}
\newtheorem{theorem}{Theorem}
\newtheorem{corollary}{Corollary}

\theoremstyle{definition}

\def\ie{{\it i.e.,}~}

\usepackage{authblk}
\usepackage{graphicx}

\usepackage{pdfpages}
\usepackage{hyperref}

\title{Fast and Compact Exact Distance Oracle for Planar Graphs}

\author{Vincent Cohen-Addad}
\author{Søren Dahlgaard\thanks{Research partly supported by Mikkel Thorup's
Advanced Grant DFF-0602-02499B from the Danish Council for Independent Research
under the Sapere Aude research career programme.}}
\author{Christian Wulff-Nilsen}
\affil{University of Copenhagen\\\texttt{[vincent.v,soerend,koolooz]@di.ku.dk}}
\date{}

\begin{document}

\setcounter{page}{0}
\maketitle
\begin{abstract}
  For a given a graph, a distance oracle is a data structure that answers
  distance queries between pairs of vertices.
  We introduce an $O(n^{5/3})$-space distance oracle 
  which answers exact distance queries in $O(\log n)$ time for $n$-vertex planar edge-weighted digraphs.
  All previous distance oracles for planar graphs with truly subquadratic space (i.e., space $O(n^{2 - \epsilon})$ 
  for some constant $\epsilon > 0$) either required query time polynomial in $n$ or could 
    only answer approximate distance queries.

  Furthermore, we show how to trade-off time and space: 
  for any $S \ge n^{3/2}$, we show how to obtain an $S$-space distance oracle that
  answers queries in time $O(\frac{n^{5/2}}{S^{3/2}} \log n)$. This is a polynomial
  improvement over the previous planar distance oracles with $o(n^{1/4})$ query
  time.

\end{abstract}

\thispagestyle{empty}

\newpage
\setcounter{page}{1}

\section{Introduction}

Efficiently storing distances between the pairs of vertices of a graph is a fundamental 
problem that has receive a lot of attention over the years.
Many graph algorithms and real-world problems require that the distances between pairs of 
vertices of a graph can be accessed efficiently. 
Given an edge-weighted digraph $G = (V,E)$ with $n$ vertices, 
a \emph{distance oracle} is a data structure that can efficiently
answer distance queries between pairs of vertices $u,v\in V$.

A naive approach consists in storing an $n\times n$ distance matrix, 
giving a distance query time of $O(1)$ by a simple table lookup. 
The obvious downside is the huge $\Theta(n^2)$ space requirement which is in 
many cases impractical.
For example, several popular routing heuristics 
(e.g.: for the travelling salesman problem) require fast access to distances 
between pairs of vertices. Unfortunately the inputs are usually too big to allow 
to store an $n \times n$ distance matrix (see e.g.:~\cite{TSPlib})\footnote{In these cases, 
the inputs are then embedded into the 2-dimensional plane so that the distances
can be computed in $O(1)$ time at the expense of working with incorrect distances.}.


Fast and compact data structures for distances are also critical in many routing problems. 
One important challenge in these applications is to process a large number of online queries
while keeping the space usage low, which is important for systems with limited
memory or memory hierarchies. 
Therefore, the alternative naive approach consisting in simply storing the graph $G$ and 
answering a query by running a shortest path algorithm on the entire graph
is also prohibitive for many applications.

Since road networks and planar graphs share many properties, 
planar graphs are often used for modeling various 
transportation networks (see e.g.:~\cite{MozesS12}).
Therefore obtaining good space/query-time trade-offs for planar distance oracles has been
studied thoroughly over the past
decades~\cite{Djidjev96,ArikatiCCDSZ96,ChenX00,FakR06,WNthesis,Cabello12,MozesS12}.

If $S$ represents the space usage and $Q$ represents the query time, the
trivial solutions described above would suggest a trade-off of 
$Q = n^2/S$\footnote{Using the $O(n)$ shortest path algorithm for planar graphs of 
Henzinger et al.~\cite{Henzinger97}.}.
Up to logarithmic factors, this trade-off is achieved by
the oracles of Djidjev~\cite{Djidjev96} and Arikati, et
al.~\cite{ArikatiCCDSZ96}. The oracle of Djidjev further improves on this
trade-off obtaining an oracle with $Q = n/\sqrt{S}$ for the range $S\in
[n^{4/3}, n^{3/2}]$ suggesting that this trade-off might instead be the correct
one. Extending this trade-off to the full range of $S$ was the subject of
several subsequent papers by Chen and Xu~\cite{ChenX00},
Cabello~\cite{Cabello12}, Fakcharoenphol and Rao~\cite{FakR06}, and finally
Mozes and Sommer~\cite{MozesS12} (see also the result of
Nussbaum~\cite{Nussbaum11}) obtaining a query time of $Q = n/\sqrt{S}$ for
the entire range of $S\in[n,n^2]$ (again ignoring constant and logarithmic
factors).

It is worth noting that the above mentioned trade-off between space usage and
query time is no better than the trivial solution of simply storing the
$n\times n$ distance matrix when constant (or even polylogarithmic) query time
is needed. In fact the best known result in this case due to
Wulff-Nilsen~\cite{WNthesis} who manages to obtain very slightly subquadratic
space of $O(n^2\polyloglog(n)/\log(n))$ and constant query time. It has been a major open question
whether an exact oracle with truly subquadratic (that is, $O(n^{2-\eps})$
for any constant $\eps > 0$) space usage and constant or even polylogarithmic query time
exists. Furthermore, the trade-offs obtained in the literature 
suggest that this might not be the case.

In this paper we break this quadratic barrier:
\begin{theorem}
  \label{Thm:main}
    Let $G = (V,E)$ be a weighted planar digraph with $n$
    vertices. Then there exists a data structure with $O(n^2)$ preprocessing time and $O(n^{5/3})$ space and a data structure with $O(n^{11/6})$ space and $O(n^{11/6})$ expected preprocessing time. Given any two query vertices $u,v\in V$, both oracles report the shortest path distance from $u$ to $v$ in $G$ in $O(\log n)$ time.
\end{theorem}
In addition to \Cref{Thm:main} we also obtain a distance oracle with a
trade-off between space
and query time.
\begin{theorem}\label{thm:tradeoff}
    Let $G = (V,E)$ be a weighted planar digraph with $n$ vertices. Let $S$
    denote the space, $P$ denote the preprocessing time, and $Q$ denote the
    query time. Then there exists planar distance oracles with the following
    properties:
    \begin{itemize}
        \item $P = O(n^2)$, $S\ge n^{3/2}$, and $Q =
            O(\frac{n^{5/2}}{S^{3/2}}\log n)$.
        \item $P = S$, $S\ge n^{16/11}$, and $Q =
            O(\frac{n^{11/5}}{S^{6/5}}\log n)$.
    \end{itemize}
\end{theorem}
In particular,
this result improves on the current state-of-the-art~\cite{MozesS12}
trade-off between space and query time for $S \ge n^{3/2}$. The main idea
is to use two $r$-divisions, where we apply our structure from \Cref{Thm:main}
to one and do a brute-force search over the boundary nodes of the other.

\paragraph{Recent developments}
We note that the main focus of this paper is on space usage and query time,
and the the preprocessing time follows directly from our proofs (and by
applying the result of \cite{Cabello17} for subquadratic time).

After posting a preliminary version of this paper on
arXiv~\cite{Cohen-AddadDW17}, 
the algorithm of Cabello~\cite{Cabello17} was
improved by Gawrychowski et al.~\cite{GawrychowskiKMSW17} to run in
$\tilde{O}(n^{5/3})$ deterministically. As noted in \cite{GawrychowskiKMSW17}
this also improves the preprocessing time of our \Cref{Thm:main} to
$\tilde{O}(n^{5/3})$ while keeping the space usage at $O(n^{5/3})$. 
It is also possible to use Gawrychowski et al. to speed-up the pre-processing time
of our distance oracle described in \Cref{thm:tradeoff}.
This yields a distance oracle (in the notation of
\Cref{thm:tradeoff}) with $P = \tilde{O}(S)$, $S\ge
n^{3/2}$, and $Q = O(\frac{n^{5/2}}{S^{3/2}}\log n)$, thus eliminating entirely
the need of the second bullet point of \Cref{thm:tradeoff}.



\subsection*{Techniques}
We derive structural results on Voronoi diagrams for planar graphs when the 
centers of the Voronoi cells lie on the same face. 
The key ingredients
in our algorithm are a novel and technical separator decomposition and point
location structure for the regions in an $r$-division allowing us to perform
binary search to find a boundary vertex $w$ lying on a shortest-path between a
query pair $u,v$. These structures are applied on top of weighted Voronoi
diagrams, and our point location structure relies heavily on partitioning
each region into small ``easy-to-handle'' wedges which are shared by many such
Voronoi diagrams. 
More high-level ideas are given in
Section~\ref{sec:HighLevel}.

Our approach bears some similarities with the recent breakthrough of Cabello~\cite{Cabello17}.
Cabello showed that abstract Voronoi diagrams~\cite{Klein89,KleinLN09} studied in computational
geometry combined with planar $r$-division can be used to obtain fast planar graphs algorithms
for computing the diameter and wiener index. 
We start from Cabello's approach of using abstract Voronoi diagrams. 
While Cabello focuses on developing fast algorithms for computing abstract Voronoi diagrams
of a planar graph, we introduce a decomposition theorem for abstract Voronoi diagrams of planar graphs
and a new data structure for point location in planar graphs.


\subsection{Related work}
In this paper, we focus on distance oracles that report shortest path distances exactly. A closely related area is approximate
distance oracles. In this case, one can obtain near-linear space and constant or near-constant query time at the cost of a small $(1+\epsilon)$-approximation factor in the distances reported~\cite{Thorup04, Klein02, kawarabayashi2011linear, kawarabayashi2013more, WN16}.

One can also study the problem in a dynamic setting, where the graph undergoes
edge insertions and deletions. Here the goal is to obtain the best trade-off between update and
query time. Fakcharoenphol and Rao~\cite{FakR06} showed how to obtain
$\Ot(n^{2/3})$ for both updates and queries and a trade-off of $O(r)$ and
$O(n/\sqrt{r})$ in general. Several follow up works have improved this result
to negative edges and shaving further logarithmic
factors~\cite{Klein05,ItalianoNSW11,KaplanMNS12,GawrK16}. Furthermore, Abboud
and Dahlgaard~\cite{AbboudD16} have showed that improving this bound to
$O(n^{1/2-\eps})$ for any constant $\eps > 0$ would imply a truly subcubic algorithm for
the All Pairs Shortest Paths (APSP) problem in general graphs.

In the seminal paper of Thorup and
Zwick~\cite{ThorupZ05}, a $(2k-1)$-approximate distance
oracle is presented for undirected edge-weighted $n$-vertex general graphs using $O(kn^{1+1/k})$ space and $O(k)$ query time for any integer $k\ge 1$. Both query time and space has subsequently been improved to $O(n^{1+1/k})$ space and $O(1)$ query time while keeping an approximation factor of $2k-1$~\cite{wulff2012approximate, Chechik14, Chechik15}. This is near-optimal, assuming the widely believed and partially proven girth conjecture of Erd{\H{o}}s~\cite{erdHos1964extremal}.

\section{Preliminaries and Notations}

Throughout this paper we denote the input graph by $G$ and we assume that it is a directed planar graph with a fixed embedding. 
We assume that $G$ is connected (when ignoring edge orientations) as otherwise each connected component can be treated separately.

Section~\ref{sec:recursivedecomp} will make use of the geometry of the 
plane and associate Jordan curves to cycle separators.
Let $H$ be a planar embedded edge-weighted digraph.
We use $V(H)$ to denote the set of
vertices of $H$ and we denote by $H^*$ 
the dual of $H$ (with parallel edges and loops) and view it as an undirected graph. We assume a natural embedding of
$H^*$ \ie each dual vertex is in the interior of its corresponding primal face and each dual edge crosses its corresponding primal edge
of $H$ exactly once and intersects no other edges of $G$. We let $d_H(u,v)$ denote the shortest path distance from vertex $u$ to vertex $v$ in $H$.

\paragraph{$r$-division}
We will rely on the notion of $r$-division introduced by 
Frederickson~\cite{Frederickson87} and further developed by
Klein et al.~\cite{KleinMS13}.
For a subgraph $H$ of $G$, a vertex $v$ of $H$ is a {\em boundary
  vertex} if $G$ contains an edge not in $H$ that is incident to $v$. We let $\delta H$ denote the set of boundary vertices of $H$. Vertices of $V(H)\setminus\delta H$ are called {\em internal vertices} of $H$.
A {\em hole} of a subgraph $H$ of $G$ is a face of $H$ that is not 
a face of $G$.

Let $c_1$ and $c_2$ be constants.  For a number $r$, an {\em
$r$-division with few holes} 
of (connected) graph $G$ (with respect to $c_1, c_2$) 
is a collection $\mathcal R$ of
subgraphs of $G$, called {\em regions}, with the following properties.
\begin{enumerate}
\itemsep0pt
\item Each edge of $G$ is in exactly one region. \label{def:rdivisiona}
\item The number of regions is at most $c_1 |V(G)|/r$.\label{def:rdivisionb}
\item Each region contains at most $r$ vertices.\label{def:rdivisionc}
\item Each region has at most $c_2\sqrt r$ boundary vertices.\label{def:rdivisiond}
\item Each region contains only $O(1)$ holes.
\end{enumerate}
We make the simplifying assumption that each hole $H$ of each region $R$ is a simple cycle and that all its vertices belong to $\delta R$. We can always reduce to this case as follows. First, turn $H$ into a simple cycle by duplicating vertices that are visited more than once in a walk of the hole. Then for each pair of consecutive boundary vertices in this walk, add a bidirected edge between them unless they are already connected by an edge of $R$; the new edges are embedded such that they respect the given embedding of $R$. We refer to the new simple cycle obtained as a hole and it replaces the old hole $H$.

We also make the simplifying assumption that each face of a region $R$ is either a hole or a triangle and that each edge of $R$ is bidirected. This can always be achieved by adding suitable infinite-weight edges that respect the current embedding of $R$.

\paragraph{Non-negative weights and unique shortest paths}
As mentioned earlier, we may assume w.l.o.g.~that $G$ has non-negative edge weights. Furthermore, we assume uniqueness of shortest paths \ie
for any two vertices $x,y$ of a graph $G$, there is a unique path from
$x$ to $y$ that minimizes the sum of the weights of its edges. This can be achieved either with random perturbations of edge weights or deterministically with a slight overhead as described in~\cite{Cabello17}; we need the shortest paths uniqueness assumption only for the preprocessing step and thus the overhead only affects the preprocessing time and not the query time of our distance oracle.



\paragraph{Voronoi diagrams}
We now define the key notion of Voronoi diagrams. 
Let $G$ be a graph and $r>0$. 
Consider an $r$-division with few holes $\calR$ of $G$ and a region 
$R \in \calR$ and let $H$
be a hole of $R$. Let $u$ be a vertex of $G$ not in $R$.

Let $R_H$ be the graph obtained from $R$ by adding inside each hole $H'\ne H$ of $R$, a new vertex in its interior and infinite-weight bidirected edges between this vertex and the vertices of $H'$ (which by the above simplifying assumption all belong to $\delta R$), embedding the edges such that they are pairwise non-crossing and contained in $H'$.

Some of the following definitions are illustrated in Figure~\ref{fig:wt_vor}. Consider the shortest path tree $T_u$ in $G$ rooted at
$u$. For any vertex $x\in V(T_u)$, define $T_u(x)$ to be the subtree of $T_u$
rooted at $x$.
For each vertex $x \in V(H)$ we define the \emph{Voronoi cell}
of $x$ (w.r.t.~$u$, $R$, and $H$)
as the set of vertices of $R_H$ that belong to the subtree $T_u(x)$ 
and not to any subtree $T_u(y)$ for $y \in (V(H) - \{x\})\cap V(T_u(x))$.
The \emph{weighted Voronoi diagram} of $u$ w.r.t.~$R$ and $H$ is the collection of all the Voronoi cells of the vertices in $H$. 
For each vertex $x \in H$ we say that its \emph{weight}
is the shortest path distance from $u$ to $x$ in $G$.
Note that since we assume unique shortest paths and 
bidirectional edges, the weighted Voronoi diagram of $u$ w.r.t.~$R$ and $H$
is a partition of the vertices of $R_H$. Furthermore, each Voronoi cell contains exactly one vertex of $H$.
For any Voronoi cell $C$, we define its boundary edges to be the edges of $R_H$ that have exactly one endpoint in $C$. Let $B_H^*$ be the subgraph of $R_H^*$ consisting of the (dual) boundary edges over all Voronoi cells w.r.t.~$u$ and $R_H$; we ignore edge orientations and weights so that $B_H^*$ is an unweighted undirected graph. We define $\text{Vor}_{H}(R,u)$ to be the multigraph
obtained from $B_H^*$ by replacing each maximal path whose interior vertices have degree two by a single edge whose embedding coincides with the path it replaces. When $H$ is clear from context, we simply write $\Vor{R}{u}$.

\begin{figure}[ht]
    \centering
    \includegraphics[width=.5\textwidth]{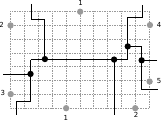}
    \caption{Illustration of the weighted Voronoi diagram of a vertex $u$ (not shown) w.r.t.~a region $R$ (grey) and a hole $H$ (dotted edges and seven grey vertices). Edges of $R$ are bidirected and have weight $1$. The
    number next to a vertex of $H$ is the weight of that vertex. Vertices of the
    Voronoi diagram (which are dual vertices and hence faces of the primal) and its boundary are shown in black except the one embedded inside $H$. The graph $\text{Vor}_{H}(R,u)$ has seven edges.
    Note that this illustration does \emph{not} have unique shortest paths as
    assumed in this paper. Note also that the shortest path from $u$ to the
    node with weight $5$ goes through the nearby node with weight $2$. Thus the
    rest of the shortest-path tree from $u$ has been assigned to the node with
    weight $5$. Note also that the illustration is not triangulated.}
    \label{fig:wt_vor}
\end{figure}

\section{High-level description}\label{sec:HighLevel}
We now give a high-level description of our distance oracle where we omit the details needed to get our preprocessing time bounds. Our data structure is constructed on top of an $r$-division of
the graph. For each region $R$ of the $r$-division we store a look-up
table of the distance in $G$ between each ordered pair of vertices $u,v\in V(R)$. We also store a
look-up table of distances in $G$ from each vertex $u\in V$ to the boundary vertices of
$R$. In total this part requires $O(nr + n^2/\sqrt{r})$ space.

The difficult case is when two vertices $u$ and $v$
from different regions are queried. To do this we will use weighted Voronoi
diagrams. More specifically, for every vertex
$u$, every region $R$, and every hole $H$ of $R$, we construct a recursive separator decomposition
of the weighted Voronoi diagram of $u$ w.r.t.~$R$ and $H$. The goal is to determine the boundary vertex $w$ such that $v$ is contained in the Voronoi cell of $u$. If we can do this,
we know that $d_G(u,v) = d_G(u,w) + d_{R_H}(w,v)$ for one of the holes $H$ of $R$. To determine this we use a
carefully selected recursive decomposition. This decomposition is stored in a compact way and we show how it enables binary search to find $w$ in $O(\log r)$ time.

In order to store all of the above mentioned parts efficiently we will employ
the compact representation of the abstract Voronoi diagram (namely $\Vor{R}{u}$ as defined in the previous section). This requires only $O(\sqrt r)$ space for each choice of $u$, $R$, and $H$ for a total of $O(n^2/\sqrt r)$ space. This also dominates the space for storing the recursive decompositions.

Finally, we store for each graph $R_H$ and each possible separator of $R_H$ the set of vertices on one side of the separator, since this is needed to perform the binary search. This is done in a compact way requiring only $O(r^2)$ space per region. Thus,
the total space requirement of our distance oracle is $O(nr + n^2/\sqrt{r})$. Picking $r=n^{2/3}$ gives the desired $O(n^{5/3})$ space bound.

\section{Recursive Decomposition of Regions}\label{sec:recursivedecomp}
In this section, we consider a region $R$ in an $r$-division of a planar
embedded graph $G = (V,E)$, a vertex $u\in V - V(R)$, and a hole $H$ of $R$
which we may assume is the outer face of $R$. To simplify notation, we identify
$R$ with $R_H$ and let $\delta R$ denote the boundary vertices of $R$ belonging
to $H$, i.e., $\delta R = V(H)$ (by our simplifying assumption regarding holes
in the preliminaries). The dual vertex corresponding to $H$ is denoted
$v_{\infty}(R,u)$ or just $v_{\infty}$.

We assume that $R$ contains at least three boundary vertices. 
Recall that each face of $R$ other than the outer face is a triangle and so,
each vertex of $\Vor{R}{u}$ other than $v_{\infty}$ has degree $3$.
Moreover, every cell
of $\Vor{R}{u}$ contains exactly one boundary vertex, therefore the
boundary of each cell of $\Vor{R}{u}$ contains exactly one occurrence of
$v_{\infty}$ and contains at least one other vertex.
Also note that the cyclic ordering of cells of
$\Vor{R}{u}$ around $v_{\infty}$ is the same as the cyclic ordering $\delta R$
of boundary vertices of $R$.

Construct a plane multigraph $\Trivor{R}{u}$ from $\Vor{R}{u}$ as follows. First, for every Voronoi cell $C$, add an edge from $v_{\infty}$ to each vertex of $C$ other than $v_{\infty}$ 
itself; 
these edges are embedded such that they are fully contained in $C$ 
and such that they are pairwise non-crossing. For each such edge $e$, denote by $C(e)$ the 
cell it is embedded in.  
To complete the construction
of $\Trivor{R}{u}$, remove every edge incident to $v_{\infty}$ 
belonging to $\Vor{R}{u}$. The construction of $\Trivor{R}{u}$ is illustrated
in Figure~\ref{fig:multigraph}.

\begin{figure}[htbp]
    \centering
    \includegraphics[width=.4\textwidth]{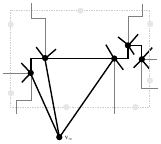}
    \caption{Illustration of $\Trivor{R}{u}$ highlighted in bold black edges.
    To avoid clutter, most edges incident to $v_\infty$ are only sketched.}
    \label{fig:multigraph}
\end{figure}




\paragraph{Recursive decomposition using a Voronoi diagram}
In this section, we show how to obtain a recursive decomposition of
$\Trivor{R}{u}$ into subgraphs called \emph{pieces}. A piece $Q$ is decomposed
into two smaller pieces by a cycle separator $S$ of size $2$ containing
$v_{\infty}$. Each of the two subgraphs of $Q$ is obtained by replacing the
faces of $Q$ on one side of $S$ by a single face bounded by $S$. The separator
$S$ is balanced w.r.t.~the number of faces of $Q$ on each side of $S$. The
recursion stops when a piece with at most six faces is obtained. It will be
clear from our construction below that the collection of cycle separators over
all recursive calls form a laminar family, i.e., they are pairwise
non-crossing.

We assume a linked list representation of each piece $Q$ where edges are ordered clockwise around each vertex.

Lemma~\ref{Lem:Separator} below shows how to find the cycle separators needed to obtain the recursive decomposition into pieces. Before we can prove it, we need the following result.
\begin{lemma}\label{Lem:PieceStructure}
Let $Q$ be one of the pieces obtained in the above recursive decomposition. Then for every vertex $v$ of $Q$ other than $v_{\infty}$,
\begin{enumerate}
    \item $v$ has at least two edges incident to $v_{\infty}$,
\item for each edge $(v,w)$ of $Q$ where $w\neq v_{\infty}$, the edge preceding and the edge following $(v,w)$ in the clockwise ordering around $v$ are both incident to $v_{\infty}$, and
\item for every pair of edges $e_1 = (v,v_{\infty})$ and $e_2 = (v,v_{\infty})$ where $e_2$ immediately follows $e_1$ in the clockwise ordering around $v$, if both edges are directed from $v$ to $v_{\infty}$ then the subset of the plane to the right of $e_1$ and to the left of $e_2$ is a single face of $Q$.
\end{enumerate}
\end{lemma}
\begin{proof}
The proof is by induction on the depth $i\ge 0$ in the recursion tree of the
    node corresponding to piece $Q$. Assume that $i = 0$ and let $v\neq
    v_{\infty}$ be given. The first and third part of the lemma follow
    immediately from the construction of $\Trivor{R}{u}$ and the assumption that $|\delta R|\ge 3$. To show the second
    part, it suffices by symmetry to consider the edge $e$ following $(v,w)$ in
    the cyclic ordering of edges around $v$ in $Q = \Trivor{R}{u}$. Since $e$
    and $(v,w)$ belong to the same face of $Q$ and since each face of $Q$ contains $v_{\infty}$ and at most three edges, the second part
    follows.

Now assume that $i > 0$ and that the claim holds for smaller values. Let $v\neq
    v_{\infty}$ be given. Consider the parent piece $Q'$ of $Q$ in the
    recursive decomposition tree and let $S$ be the cycle separator that was
    used to decompose $Q'$. Then $S$ contains $v_{\infty}$ and one additional
    vertex $v'$. To show the inductive step, we claim that we only need to consider the case
    when $v' = v$. This is clear for the first and second part since if $v'\ne v$ then $v$ has the same set of incident edges in $Q'$ and in $Q$. It is also clear for the third part since $Q$ is a subgraph of $Q'$.

It remains to show the induction step when $v' = v$. The first part follows since the two edges of $S$ are incident to $v$ and to $v_{\infty}$ and belong to $Q$. The second part follows by observing that the clockwise ordering of edges around $v$ in $Q$ is obtained from the clockwise ordering around $v$ in $Q'$ by removing an interval of consecutive edges in this ordering; furthermore, the first and last edge in the remaining interval are both incident to $v_{\infty}$. Applying the induction hypothesis shows the second part.

For the third part, if $e_2$ immediately follows $e_1$ in the clockwise ordering around $v$ in $Q'$ then the induction hypothesis gives the desired. Otherwise, $e_1$ and $e_2$ must be the two edges of $S$ and $Q$ is obtained from $Q'$ by removing the faces to the right of $e_1$ and to the left of $e_2$ and replacing them by a single face bounded by $S$.
\end{proof}

In the following, let $Q$ be a piece with more than six faces. The following lemma shows that $Q$ has a balanced cycle separator of size $2$ which can be found in $O(|Q|)$ time.
\begin{lemma}\label{Lem:Separator}
$Q$ as defined above contains a $2$-cycle $S$ containing $v_{\infty}$ such that the number of faces of $Q$ on each side of $S$ is a fraction between $1/3$ and $2/3$ of the total number of faces of $Q$. Furthermore, $S$ can be found in $O(|Q|)$ time.
\end{lemma}
\begin{proof}
We construct $S$ iteratively. In the first iteration, pick an arbitrary vertex $v_1\neq v_{\infty}$ of $Q$ and let $S_1$ consist of two distinct arbitrary edges, both incident to $v_1$ and $v_{\infty}$. This is possible by the first part of Lemma~\ref{Lem:PieceStructure}.

Now, consider the $i$th iteration for $i > 1$ and let $v_{i-1}$ and $v_{\infty}$ be the two vertices of the $2$-cycle $S_{i-1}$ obtained in the previous iteration. If $S_{i-1}$ satisfies the condition of the lemma, we let $S = S_{i-1}$ and the iterative procedure terminates.

Otherwise, one side of $S_{i-1}$ contains more than $2/3$ of the faces of $Q$. Denote this set of faces by $\mathcal F_{i-1}$ and let $E_{i-1}$ be the set of edges of $Q$ incident to $v_{i-1}$, contained in faces of $\mathcal F_{i-1}$, and not belonging to $S_{i-1}$. We must have $E_{i-1}\neq\emptyset$; otherwise, it follows from the third part of Lemma~\ref{Lem:PieceStructure} that $\mathcal F_{i-1}$ contains only a single face of $Q$ (bounded by $S_{i-1}$), contradicting our assumption that $Q$ contains more than six faces and that $\mathcal F_{i-1}$ contains more than $2/3$ of the faces of $Q$.

If $E_{i-1}$ contains an edge incident to $v_{\infty}$, pick an arbitrary such edge $e_{i-1}$. This edge partitions $\mathcal F_{i-1}$ into two non-empty subsets; let $\mathcal F_{i-1}'$ be the larger subset. We let $v_i = v_{i-1}$ and let $S_i$ be the $2$-cycle consisting of $e_{i-1}$ and the edge of $S_{i-1}$ such that one side of $S_i$ contains exactly the faces of $\mathcal F_{i-1}'$.

Now, assume that none of the edges of $E_{i-1}$ are incident to $v_{\infty}$. Then by the second part of Lemma~\ref{Lem:PieceStructure}, $E_{i-1}$ contains exactly one edge $e_{i-1}$. We let $v_i\neq v_{i-1}$ be the other endpoint of $e_{i-1}$ and we let $S_i$ consist of the two edges incident to $v_i$ which belong to the two faces of $\mathcal F_{i-1}$ incident to $e_{i-1}$.

To show the first part of the lemma, it suffices to prove that the above iterative procedure terminates. Consider two consecutive iterations $i > 1$ and $i+1$ and assume that the procedure does not terminate in either of these. We claim that then $\mathcal F_i\subset \mathcal F_{i-1}$. If we can show this, it follows that $|\mathcal F_1| > |\mathcal F_2| >\ldots$ which implies termination.

If $E_{i-1}$ contains an edge incident to $v_{\infty}$ then $\mathcal F_{i-1}'$ contains more than $1/3$ of the faces of $Q$. Since the procedure does not terminate in iteration $i+1$, $\mathcal F_{i-1}'$ must in fact contain more than $2/3$ of the faces so $\mathcal F_i = \mathcal F_{i-1}'\subset \mathcal F_{i-1}$, as desired.

Now, assume that none of the edges of $E_{i-1}$ are incident to $v_{\infty}$. Then one side of $S_i$ contains exactly the faces of $\mathcal F_{i-1}$ excluding two. Since we assumed that $Q$ contains more than six faces, this side of $S_i$ contains more than $2/3$ of these faces. Hence, $\mathcal F_i\subset\mathcal F_{i-1}$, again showing the desired.

For the second part of the lemma, note that counting the number of faces of $Q$ contained in one side of a $2$-cycle containing $v_{\infty}$ can be done in the same amount of time as counting the number of edges incident to $v_{\infty}$ from one edge of the cycle to the other in either clockwise or counter-clockwise order around $v_{\infty}$. This holds since every face of $Q$ contains $v_{\infty}$. It now follows easily from our linked list representation of $Q$ with clockwise orderings of edges around vertices that the $i$th iteration can be executed in $O(|\mathcal F_{i-1}| - |\mathcal F_i|)$ time for each $i > 1$. This shows the second part of the lemma.
\end{proof}
\begin{corollary}\label{Cor:Separatortime}
  Given $\Vor{R}{u}$ and $\Trivor{R}{u}$, its recursive decomposition can
  be computed in $O(\sqrt r\log r)$ time.
\end{corollary}
\begin{proof}
$\Vor{R}{u}$ has complexity $|\Vor{R}{u}| = O(\sqrt r)$ and $\Trivor{R}{u}$ can be found in time linear in this complexity. Since the recursive decomposition of $\Trivor{R}{u}$ has $O(\log r)$ levels and since the total size of pieces on any single level is $O(|\Trivor{R}{u}|) = O(|\Vor{R}{u}|) = O(\sqrt r)$, the corollary follows from Lemma~\ref{Lem:Separator}.
\end{proof}


\begin{lemma}
  \label{Lem:sepspace}
  The recursive decomposition of $\Trivor{R}{u}$ can be stored using $O(\sqrt r)$ space.
\end{lemma}
\begin{proof}
  Observe that the number of nodes of the tree decomposition is 
  $O(\sqrt{r})$ and each separator consists of two edges and so takes
  $O(1)$ space.
\end{proof}

\paragraph{Embedding of $\Trivor{R}{u}$:}
We now provide a more precise definition of the embedding of $\Trivor{R}{u}$.
Let $f_{\infty}$ be the face of $R$ corresponding to $v_{\infty}$ in $R^*$, i.e., $f_{\infty}$ is the hole $H$.
Consider the graph $\tilde{R}$ that consists of $R$ plus
a vertex $\tilde{v}_{\infty}$ located in $f_{\infty}$ and an edge between
each vertex of $f_{\infty}$ and $\tilde{v}_{\infty}$. The rest of $\tilde{R}$ is 
embedded consistently with respect to the embedding of 
$R$.

Now, consider the following embedding of $\Trivor{R}{u}$. 
First, embed $v_{\infty}$ to $\tilde{v}_{\infty}$. We now specify the embedding of each edge
adjacent to $v_{\infty}$.
Recall that each edge $e$ that is adjacent to $v_{\infty}$ lies in 
a single cell $C(e)$ of $\Vor{R}{u}$.
For each such edge $e$ going from $v_{\infty}$ to a vertex $w^*$ of $\Trivor{R}{u}$, we 
embed it so that it follows the edge from $\tilde{v}_{\infty}$ to the boundary 
vertex $b_e$ of $C(e)$,
then the shortest path in $\tilde{R}$ from $b_e$ to the vertex of $C(e)$ on the face 
corresponding to $w^*$ in $\tilde{R}$. Note that by definition of $\Trivor{R}{u}$ such 
a vertex exists. We also remark that since the edges follow shortest 
paths and because of the uniqueness of the shortest paths
they may intersect but not cross (and hence do not contradict 
the definition of $\Trivor{R}{u}$).
  
It follows that there exists a 1-to-1 correspondence between 2-cycle separators going through
$v_{\infty}$ of $\Trivor{R}{u}$ and cycle separators of $\tilde{R}$ consisting of 
an edge $(u,v)$, the shortest paths between $u$ and a boundary vertex $b_1$ and $v$ 
and a boundary vertex $b_2$ and $(b_1,\tilde{v}_\infty)$ and
$(b_2,\tilde{v}_\infty)$. We call the set $\{b_1,u,v,b_2\}$ the \emph{representation} 
of this separator.
This is illustrated in
Figure~\ref{fig:sep2}.
Thus, for any 2-cycle separator $S$ going through $v_{\infty}$, 
we say that the set of vertices of $R$ that is in the interior (resp. exterior) of $S$ 
is the set of vertices of $R$ that lie in the bounded region of the 
place defined by the Jordan curve corresponding to 
the cycle separator in $\tilde{R}$ that is in 1-1 correspondence 
with $S$.

\begin{figure}[htbp]
    \centering
    \includegraphics[width=.5\textwidth]{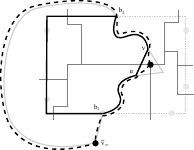}
    \caption{Example of a $2$-cycle separator of $\Trivor{R}{u}$ and its
    corresponding embedding into the actual graph.}
    \label{fig:sep2}
\end{figure}

We can now state the main lemma of this section.

\begin{lemma}
  \label{lem:mainseparator}
  Let $w$ be a vertex of $R$. Assume there exists a data structure that 
  takes as input a the representation $\{b_1,x,y,b_2\}$ 
  of a 2-cycle separator $S$ of 
  $\Vor{R}{u}$ going through $v_{\infty}$ 
  and answers in $t$ time queries of 
  the following form: Is $w$ in the bounded closed subset of the plane
  with boundary $S$?
  Then there exists an algorithm running in time $O(t \log r)$ 
  that returns 
  a set of at most $6$ Voronoi cells of $\Vor{R}{u}$ such that 
  one of them contains $w$.

\end{lemma}
\begin{proof}
  The algorithm uses the recursive decomposition of $\Trivor{R}{u}$
  described in this section. Note
  that the decomposition consists of 2-cycle separators 
  going through $v_{\infty}$. Thus, using the above embedding, each of the 2-cycle
  of the decomposition corresponds to a separator
  consisting of an edge $(x,y)$ and the shortest paths $P_R(x,b_1)$ and 
  $P_R(y,b_2)$ where $b_1,b_2$ are boundary vertices of $R$.
  Additionally, $y$ belongs to the Voronoi cell of $b_2$ in $\Vor{R}{u}$ and
  $x$ belongs to the Voronoi cell of $b_1$ in $\Vor{R}{u}$. Therefore,
  $P_R(x,b_1)$ and $P_R(y,b_2)$ are vertex disjoint and so the data structure can
  be used to decide on which side of such a separator $w$ is.

  The algorithm is the following: proceed recursively along the recursive
  decomposition of $\Trivor{R}{u}$ and for each $2$-cycle separator of the 
  decomposition use the data structure to decide in $t$ time in which side
  of the 2-cycle $w$ is located and then recurse on this side. If $w$ belongs to both sides, i.e., if $w$ is on the $2$-cycle separator, recurse on an arbitrary side. The algorithm stops when there are at most $6$ faces
  of $\Trivor{R}{u}$ and then it returns the Voronoi cells of $\Vor{R}{u}$ 
  intersecting those $6$ faces.

  Observe that the separators do not cross. 
  Thus, when the algorithm obtains
  at a given recursive call that $w$ is in the interior (resp. exterior) 
  of a 2-cycle $S$ and in the exterior (resp. interior) of the 2-cycle separator
  $S'$ corresponding to the next recursive call, we can deduce that $w$ 
  lies in the intersection of the interior of $S$ and the exterior of $S'$ and
  hence deduce that it belongs to a Voronoi cell that lies in this area of the 
  plane.

  Note that by Lemma~\ref{Lem:Separator} the number of faces of $\Trivor{R}{u}$
  in a piece decreases by a constant factor at each step. Thus, since the number
  of boundary vertices is  $O(\sqrt{r})$, the procedure takes 
  at most $O(t \log r)$ time.
  
  Finally, observe that each face of $\Trivor{R}{u}$ 
  that is adjacent to $v_{\infty}$ lies in a single Voronoi cell of $\Vor{R}{u}$. 
  Thus, since at the end of the recursion there are at most $6$ faces in the piece, 
  they correspond to at most $6$ different Voronoi cells of $\Vor{R}{u}$. Hence the 
  algorithm returns at most $6$ different Voronoi cells of $\Vor{R}{u}$.

\end{proof}

\section{Preprocessing a Region}\label{sec:PreprocRegion}
Given a query separator $S$ in a graph $R_H = R$ and given a query vertex $w$ in $R$, our data structure needs to determine in $O(1)$ time the side of $S$ that $w$ belongs to. In this section, we describe the preprocessing needed for this.

In the following, fix $R$ as well as an ordered pair $(u,v)$ of vertices of $R$ such that either $(u,v)$ or $(v,u)$ is an edge of $R$. The preprocessing described in the following is done over all such choices of $R$ and $(u,v)$ (and all holes $H$).

The vertices of $\delta R$ are on a simple cycle and we identify $\delta R$ with this cycle which we orient clockwise (ignoring the edge orientations of $R$). We let $b_0,\ldots,b_k$ denote this clockwise ordering where $b_k = b_0$. It will be convenient to calculate indices modulo $k$ so that, e.g., $b_{k+1} = b_1$.

Given vertices $w$ and $w'$ 
in $R$, let $P(w,w')$ denote the shortest path in $R$ from $w$ to $w'$. Given two vertices $b_i,b_j\in\delta R$, we let $\delta(b_i,b_j)$ denote the subpath of cycle $\delta R$ consisting of the vertices from $b_i$ to $b_j$ in clockwise order, where $\delta(b_i,b_j)$ is the single vertex $b_i$ if $i = j$ and $\delta(b_i,b_j) = \delta R$ if $j = i + k$. We let $\Delta(w,b_i,b_j)$ denote the subgraph of $R$ contained in the closed and bounded region of the plane with boundary defined by $P(b_i,w)$, $P(b_j,w)$, and $\delta(b_i,b_j)$. We refer to $\Delta(w,b_i,b_j)$ as a \emph{wedge} and call it a \emph{basic wedge} if $b_i$ and $b_j$ are consecutive in the clockwise order, i.e., if $j = i+1\pmod k$. We need the following lemma.
\begin{lemma}\label{Lem:PieQuery}
Let $w$ be a given vertex of $R$. Then there is a data structure with $O(r)$ preprocessing time and size which answers in $O(1)$ time queries of the following form: given a vertex $x\in V(R)$ and two distinct vertices $b_{i_1},b_{i_2}\in\delta R$, does $x$ belong to $\Delta(w,b_{i_1},b_{i_2})$?
\end{lemma}
\begin{proof}
Below we present a data structure with the bounds in the lemma which only answers restricted queries of the form ``does $x$ belong to $\Delta(w,b_0,b_i)$?'' for query vertices $x\in V(R)$ and $b_i\in\delta R$. In a completely symmetric manner, we obtain a data structure for restricted queries of the form ``does $x$ belong to $\Delta(w,b_i,b_k)$?'' for query vertices $x\in V(R)$ and $b_i\in\delta R$. We claim that this suffices to show the lemma. For consider a query consisting of $x\in V(R)$ and $b_i,b_j\in\delta R$. If $b_0\in\delta(b_i,b_j)$ then $\Delta(w,b_i,b_j) = \Delta(w,b_i,b_k)\cup\Delta(w,b_0,b_j)$ and otherwise, $\delta(w,b_i,b_j) = \Delta(w,b_0,b_j)\cap\Delta(w,b_i,b_k)$. Hence, answering a general query can be done using two restricted queries and checking if $b_0\in\delta(b_1,b_2)$ can be done in constant time by comparing indices of the query vertices.

It remains to present the data structure for restricted queries of the form ``does $x$ belong to $\Delta(w,b_0,b_i)$?''. In the preprocessing step, each $v\in V(R)$ is assigned the smallest index $i_v\in\{0,\ldots,k\}$ for which $v\in\Delta(w,b_0,b_{i_v})$. Clearly, this requires only $O(r)$ space and below we show how to compute these indices in $O(r)$ time.

Consider a restricted query specified by a vertex $x$ of $R$ and a boundary vertex $b_i\in\delta R$ where $0\le i\le k$. Since $x\in\Delta(w,b_0,b_i)$ iff $i_x\le i$, this query can clearly be answered in $O(1)$ time.

It remains to show how the indices $i_v$ can be computed in a total of $O(r)$ time. Let $R'$ be $R$ with all its edge directions reversed. In $O(r)$ time, a SSSP tree $T'$ from $w$ in $R'$ is computed. Let $T$ be the tree in $R$ obtained from $T'$ by reversing all its edge directions; note that all edges of $T$ are directed towards $w$ and for each $v\in V(R)$, the path from $v$ to $w$ in $T$ is a shortest path from $v$ to $w$ in $R$.

Next, $\Delta(w,b_0,b_0) = P(b_0,w)$ is computed and for each vertex $v\in\Delta(w,b_0,b_0)$, set $i_v = 0$. The rest of the preprocessing algorithm consists of iterations $i = 1,\ldots,k$ where iteration $i$ assigns each vertex $v\in V(\Delta(w,b_0,b_i))\setminus V(\Delta(w,b_0,b_{i-1}))$ the index $i_v = i$. This correctly computes indices for all vertices of $R$. In the following, we describe how iteration $i$ is implemented.

First, the path $P(b_i,w)$ is traversed in $T$ until a vertex $v_i$ is encountered which previously received an index. In other words, $v_i$ is the first vertex on $P(b_i,w)$ belonging to $\Delta(w,b_0,b_{i-1})$. Note that $v_i$ is well-defined since $w\in\Delta(w,b_0,b_{i-1})$. Vertices that are in $V(P(b_i,v_i))\setminus\{v_i\}$ or in a subtree of $T$ rooted in a vertex of $V(P(b_i,v_i))\setminus\{v_i\}$ and extending to the right of this path are assigned the index value $i$. Furthermore, vertices belonging to a subtree of $T$ rooted in a vertex of $V(P(b_{i-1},v_i))\setminus\{v_i\}$ and extending to the left of this path are assigned the index value $i$, except those on $P(b_{i-1},v_i)$ (as they belong to $\Delta(w,b_0,b_{i-1})$).

Since $R$ is connected, it follows that the vertices assigned an index of $i$ are exactly those belonging to $V(\Delta(w,b_0,b_i))\setminus V(\Delta(w,b_0,b_{i-1}))$ and that the running time for making these assignments is $O(|V(\Delta(w,b_0,b_i))\setminus V(\Delta(w,b_0,b_{i-1}))| + |P(b_{i-1},v_i) - v_i| + 1)$. Over all $i$, total running time is $O(r)$; this follows by a telescoping sums argument and by observing that vertex sets $V(P(b_{i-1},v_i))\setminus\{v_i\}$ are pairwise disjoint.
\end{proof}

Given distinct vertices $b_{i_1},b_{i_2}\in\delta R$, if $P(b_{i_1},u)$ and $P(b_{i_2},v)$ do
not cross (but may touch and then split), let $\Box(b_{i_1},b_{i_2},u,v)$ denote the subgraph of
$R$ contained in the closed and bounded region of the plane with boundary
defined by $P(b_{i_1},u)$, $P(b_{i_2},v)$, $\delta(b_{i_1},b_{i_2})$, and an edge of $R$
between vertex pair $(u,v)$. In order to simplify notation, we shall omit $u$ and $v$ and simply write $\Box(b_{i_1},b_{i_2})$.

It follows from planarity that there is at most one $b_i\in\delta R$ such that $(u,v)$ belongs to $E(\Delta(u,b_i,b_{i+1}))\setminus E(P(b_{i+1},u))$ when ignoring edge orientations. If $b_i$ exists, we refer to it as
$b_{uv}$; otherwise $b_{uv}$ denotes some dummy vertex not belonging to $R$.

The goal in this section is to determine whether a given query vertex belongs to a given query subgraph $\Box(b_{i_1},b_{i_2})$. The following lemma allows us to decompose this subgraph into three simpler parts as illustrated in
Figure~\ref{fig:canonical}. We will show how to answer containment queries for
each of these simple parts.
\begin{figure}[htbp]
    \centering
    \includegraphics[width=.5\textwidth]{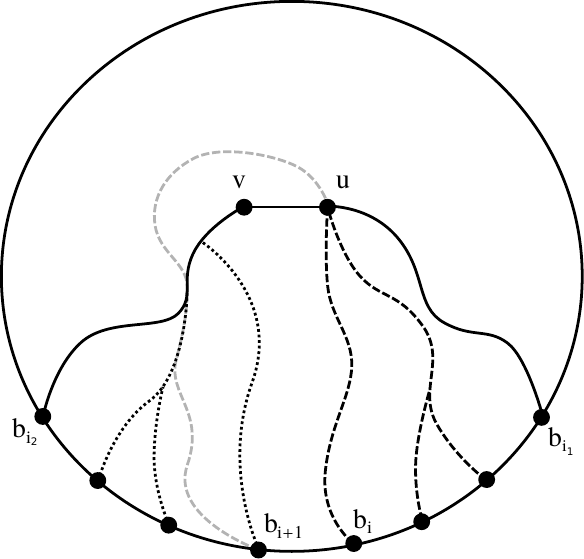}
    \caption{Example of decomposing the region into three parts using $b_i$:
    The dashed wedges represent shortest paths to $u$, the dotted edges
    represent shortest paths to $v$, and the dashed-dotted box
    $\Box(b_i,b_{i+1})$.}
    \label{fig:canonical}
\end{figure}

\begin{lemma}\label{Lem:SepDecomposition}
Let $b_{i_1}$ and $b_{i_2}$ be distinct vertices of $\delta R$ and assume that
    $P(b_{i_1},u)$ and $P(b_{i_2},v)$ are vertex-disjoint. Then $\Box(b_{i_1},b_{i_2}) = \Delta(u,b_{i_1},b_i)\cup\Box(b_i,b_{i+1})\cup\Delta(v,b_{i+1},b_{i_2})$ where $b_i = b_{uv}$ if $b_{uv}\in\delta(b_{i_1},b_{i_2-1})$ and $b_i = b_{i_2-1}$ otherwise.
\end{lemma}
\begin{proof}
Figure~\ref{fig:sepProof} gives an illustration of the proof. We first show the following result: given a vertex $b_j\in\delta(b_{i_1+1},b_{i_2})$ such that $b_{uv}\notin\delta(b_{i_1},b_{j-1})$, $P(b_{j'},u)$ is contained in $\Box(b_{i_1},b_{i_2})$ for each $b_{j'}\in\delta(b_{i_1},b_j)$. The proof is by induction on the number $i$ of edges in $\delta(b_{i_1},b_{j'})$. The base case $i = 0$ is trivial since then $b_{j'} = b_{i_1}$ so assume that $i > 0$ and that the claim holds for $i - 1$. If $(v,u)$ is the last edge on $P(b_{j'},u)$, the induction step follows from uniqueness of shortest paths. Otherwise, neither $(u,v)$ nor $(v,u)$ belong to $\Delta(u,b_{j'-1},b_{j'})$ (since $b_{j'-1}\ne b_{uv}$). By the induction hypothesis, $P(b_{j'-1},u)$ is contained in $\Box(b_{i_1},b_{i_2})$ and since $P(b_{j'},u)$ cannot cross $P(b_{j'-1},u)$, $P(b_{j'},u)$ cannot cross $P(b_{i_1},u)$. Also, $P(b_{j'},u)$ cannot cross $P(b_{i_2},v)$ since then either $(u,v)$ or $(v,u)$ would belong to $\Delta(u,b_{j'-1},b_{j'})$. Since $b_{j'}\notin\{b_{i_1+1},b_{i_2-1}\}$, it follows that $P(b_{j'},u)$ is contained in $\Box(b_{i_1},b_{i_2})$ which completes the proof by induction.

Next, assume that $b_{uv}\notin\delta(b_{i_1},b_{i_2-1})$ so that $b_i = b_{i_2-1}$. Note that $\Delta(v,b_{i+1},b_{i_2}) = P(b_{i_2},v)$. Picking $b_j = b_{i_2}$ above implies that $\Delta(u,b_{i_1},b_i)$ is contained in $\Box(b_{i_1},b_{i_2})$ and hence $\Box(b_{i_1},b_{i_2}) = \Delta(u,b_{i_1},b_i)\cup\Box(b_i,b_{i+1})\cup\Delta(v,b_{i+1},b_{i_2})$.


Now consider the other case of the lemma where $b_i = b_{uv}\in\delta(b_{i_1},b_{i_2-1})$. Picking $b_j = b_i$ above, it follows that $\Delta(u,b_{i_1},b_i)$ is contained in $\Box(b_{i_1},b_{i_2})$. It suffices to show that $P(b_{i+1},v)$ is contained in $\Box(b_{i_1},b_{i_2})$ since this will imply that $\Box(b_i,b_{i+1})$ is well-defined and that $\Box(b_i,b_{i+1})\cup\Delta(v,b_{i+1},b_{i_2})$ is contained in $\Box(b_{i_1},b_{i_2})$ and hence that $\Delta(u,b_{i_1},b_i)\cup\Box(b_i,b_{i+1})\cup\Delta(v,b_{i+1},b_{i_2}) = \Box(b_{i_1},b_{i_2})$.

Assume for contradiction that $P(b_{i+1},v)$ is not contained in $\Box(b_{i_1},b_{i_2})$.

By uniqueness of shortest paths, $P(b_{i+1},u)$ does not cross $P(b_i,u)$. Since $b_i = b_{uv}$, we have that when ignoring edge orientations, $(u,v)$ belongs to $E(\Delta(u,b_i,b_{i+1}))\setminus E(P(b_{i+1},u))\subseteq E(\Delta(u,b_{i_1},b_{i+1}))\setminus E(P(b_{i+1},u))$. Hence $P(b_{i+1},u)$ is not contained in $\Box(b_{i_1},b_{i_2})$ so it crosses $P(b_{i_2},v)$. Let $x$ be a vertex on $P(b_{i+1},u)\cap P(b_{i_2},v)$ such that the successor of $x$ on $P(b{i+1},u)$ is not contained in $\Box(b_{i_1},b_{i_2})$. By our assumption above that $P(b_{i+1},v)$ is not contained in $\Box(b_{i_1},b_{i_2})$, there is a first vertex $y$ on $P(b_{i+1},v)$ such that its successor $y'$ does not belong to $\Box(b_{i_1},b_{i_2})$. By uniqueness of shortest paths, $y$ cannot belong to $P(b_{i_2},v)$ so it must belong to $P(b_{i_1},u)$. This also implies that $y\neq x$ since $x\in P(b_{i_2},v)$ and $P(b_{i_1},u)$ and $P(b_{i_2},v)$ are vertex-disjoint. Since $P(y,v)$ is a subpath of $P(b_{i+1},v)$ and $y\neq x$, shortest path uniqueness implies that $P(y,v)$ and $P(b_{i+1},u)$ are vertex-disjoint.

Since $P(b_{i+1},y)$ is contained in $\Box(b_{i_1},b_{i_2})$, $v$ belongs to the subgraph of $\Delta(u,b_{i_1},b_{i+1})$ contained in the closed region of the plane bounded by $P(b_{i+1},y)$, $P(y,u)$, and $P(b_{i+1},u)$. Since $P(y,v)$ does not intersect $P(b_{i+1},u)$, $P(y',v)$ thus intersects either $P(b_{i+1},y)$ or $P(y,u)$. However, it cannot intersect $P(b_{i+1},y)$ since then $P(b_{i+1},v)$ would be non-simple. By uniqueness of shortest paths, $P(y',v)$ also cannot intersect $P(y,u)$ since $P(y',v)$ is a subpath of $P(y,v)$ and $y'\notin P(y,u)$. This gives the desired contradiction, concluding the proof.
\end{proof}
\begin{figure}[htbp]
    \centering
    \includegraphics[width=.5\textwidth]{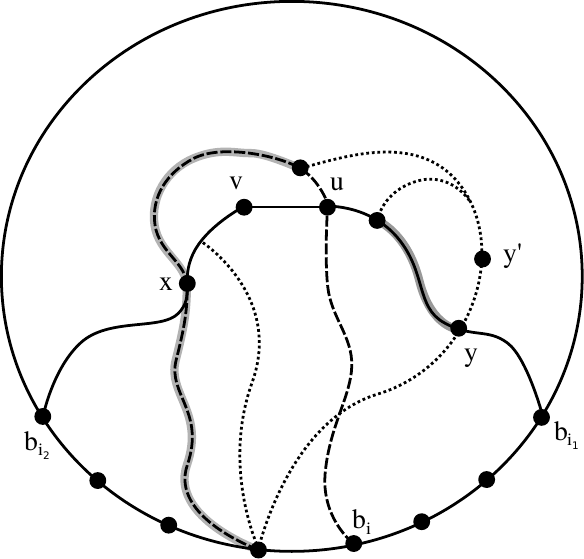}
    \caption{Illustration of the proof of Lemma~\ref{Lem:SepDecomposition}. The
    figure shows how unique shortest paths imply a contradiction (highlighted
    with grey) if the path
    from $b_{i+1}$ to $v$ is not contained in $\Box(b_{i_1},b_{i_2})$.}
    \label{fig:sepProof}
\end{figure}

Let $\mathcal P$ be a collection of subpaths such that for each path $\delta(b_{i_1},b_{i_2})$ in $\mathcal P$, $b_{uv}\notin\delta(b_{i_1},\ldots,b_{i_2-1})$ and $b_{vu}\notin\delta(b_{i_1+1},\ldots,b_{i_2-1})$. We may choose the paths such that $|\mathcal P| = O(1)$ and such that all edges of $\delta R$ except $(b_{uv},b_{vu})$ (if it exists) belongs to a path of $\mathcal P$. It is easy to see that this is possible by considering a greedy algorithm which in each step picks a maximum-length path which is edge-disjoint from previously picked paths and which satisfies the two stated requirements.

The next lemma allows us to obtain a compact data structure to answer queries of the form ``does face $f$ belong to $\Box(b_i,b_{i+1})$'' for given query face $f$ and query index $i$.
\begin{lemma}\label{Lem:CompressedStrips}
Let $P = \delta(b_{i_1},b_{i_2})\in\mathcal P$ be given. Then
\begin{enumerate}
\item an index $j(P)$ exists with $i_1\le j(P)\le i_2$ such that $\Box(b_i,b_{i+1})$ is undefined for $i_1\le i < j(P)$ and well-defined for $j(P)\le i < i_2$,
\item for each face $f\ne\delta R$ of $R$, there is at most one index $j_f(P)$ with $j(P)\le j_f(P)\le i_2 - 2$ such that $f\subseteq\Box(b_{j_f(P)},b_{j_f(P)+1})$ and $f\nsubseteq\Box(b_{j_f(P)+1},b_{j_f(P)+2})$, and
\item for each face $f\ne\delta R$ of $R$, there is at most one index $j_f'(P)$ with $j(P)\le j_f'(P)\le i_2 - 2$ such that $f\nsubseteq\Box(b_{j_f'(P)},b_{j_f'(P)+1})$ and $f\subseteq\Box(b_{j_f'(P)+1},b_{j_f'(P)+2})$.
\end{enumerate}
Furthermore, there is an algorithm which computes the index $j(P)$ and for each face $f\ne\delta R$ of $R$ the indices $j_f(P)$ and $j_f'(P)$ if they exist. The total running time of this algorithm is $O(r\log r)$ and its space requirement is $O(r)$.
\end{lemma}
\begin{proof}
Let path $P = \delta(b_{i_1},b_{i_2})\in\mathcal P$ and face $f\ne\delta R$ of $R$ be given. To simplify notation in the proof, we shall omit reference to $P$ and write, e.g., $j$ instead of $j(P)$.

Let $j$ be the smallest index such that $\Box(b_j,b_{j+1})$ is well-defined; if $j$ does not exist, pick instead $j = i_2$. We will show that $j$ satisfies the first part of the lemma. This is clear if $j = i_2$ so assume therefore in the following that $j < i_2$.

We prove by induction on $i$ that $\Box(b_i,b_{i+1})$ is well-defined for $j\le i < i_2$. By definition of $j$, this holds when $i = j$. Now, consider a well-defined subgraph $\Box(b_i,b_{i+1})$ where $j\le i \le i_2 - 2$. We need to show that $\Box(b_{i+1},b_{i+2})$ is well-defined. Since $b_i\ne b_{uv}$, $P(b_{i+1},u)$ is contained in $\Box(b_i,b_{i+1})$ and since $b_{i+1}\ne b_{vu}$, $P(b_{i+2},v)$ is contained in the closed region of the plane bounded by the boundary of $\Box(b_i,b_{i+1})$ and not containing $\Box(b_i,b_{i+1})$. In particular, $\Box(b_{i+1},b_{i+2})$ is well-defined. This shows the first part of the lemma.

Next, we show that $j$ can be computed in $O(r\log r)$ time. Checking that $\Box(b_i,b_{i+1})$ is well-defined (i.e., that paths $P(b_i,u)$ and $P(b_{i+1},v)$ do not cross) for a given index $i$ can be done in $O(r)$ time. Because of the first part of the lemma, a binary search algorithm can be applied to identify $j$ in $O(\log r)$ steps where each step checks if $\Box(b_i,b_{i+1})$ is well-defined for some index $i$. This gives a total running time of $O(r\log r)$, as desired. Space is clearly $O(r)$.

To show the second part of the lemma, assume that there is an index $j_f$ with $i_1\le j_f\le i_2 - 2$ such that $f\subseteq\Box(b_{j_f},b_{j_f+1})$ and $f\nsubseteq\Box(b_{j_f+1},b_{j_f+2})$. It follows from the observations in the inductive step above that $\Delta(u,b_{j_f},b_{j_f+1})$ contains exactly the faces of $R$ contained in $\Box(b_{j_f},b_{j_f+1})$ and not in $\Box(b_{j_f+1},b_{j_f+2})$ which implies that $f\subseteq\Delta(u,b_{j_w},b_{j_w+1})$. Since no face of $R$ belongs to more than one graph of the form $\Delta(u,b_i,b_{i+1})$, $j_f$ must be unique, showing the second part of the lemma.

Next, we give an $O(r)$ time and space algorithm that computes indices $j_f$. Let $T$ be constructed as in the proof of Lemma~\ref{Lem:PieQuery}. Initially, vertices of $P(b_{i_2-1},u)$ are marked and all other vertices of $R$ are unmarked. The remaining part of the algorithm consists of iterations $i_2 - 2,\ldots,j$ in that order. In iteration $i$, $P(b_i,u)$ is traversed until a marked vertex $v_i$ is visited and then the vertices of $P(b_i,v_i)$ are marked. The faces of $R$ contained in the bounded region of the plane defined by $P(b_i,v_i)$, $P(b_{i+1},v_i)$, and $\delta(b_i,b_{i+1})$ are exactly those that should be given an index value of $i$. The algorithm performs this task by traversing each subtree of $T$ emanating to the right of $P(b_i,v_i)$ and each subtree of $T$ emanating to the left of $P(b_{i+1},v_i)$; for each vertex visited, the algorithm assigns the index value $i$ to its incident faces.

We now show that the algorithm for computing indices $j_f$ has $O(r)$ running time. Using the same arguments as in the proof of Lemma~\ref{Lem:PieQuery}, the total time to traverse and mark paths $P(b_i,v_i)$ is $O(r)$. The total time to assign indices to faces is $O(r)$; this follows by observing that the time spent on assigning indices to faces incident to a vertex of $T$ is bounded by its degree and this vertex is not visited in other iterations.

The third part of the lemma follows with essentially the same proof as for the second part.
\end{proof}

We can now combine the results of this section to obtain the data structure described in the following lemma.
\begin{lemma}\label{Lem:RegionDS}
Let $(u,v)$ be a vertex pair connected by an edge in $R$. Then there is a data structure with $O(r\log r)$ preprocessing time and $O(r)$ space which answers in $O(1)$ time queries of the
    following form: given a vertex $w\in R$ and two distinct vertices $b_i,
    b_j\in\delta R$ such that $P(b_i,u)$ and $P(b_j,v)$ are vertex-disjoint, does $w$
    belong to $\Box(b_i,b_j,u,v)$?
\end{lemma}
\begin{proof}
We present a data structure $\mathcal D(u,v)$ satisfying the lemma. First we
    focus on the preprocessing. Boundary vertices $b_{uv}$ and $b_{vu}$ and set
    $\mathcal P$ as defined above are precomputed and stored. Vertices of $\delta R$ are labeled with indices $b_0,\ldots,b_{|V(\delta R)|-1}$ according to a clockwise walk of $\delta R$. Each path of the form $\delta(b_{i_1},b_{i_2})$ (including each path in $\mathcal P$) is represented by the ordered index pair $(i_1,i_2)$. Checking if a given boundary vertex belongs to such a given path can then be done in $O(1)$ time.

Next, two instances of the data structure in
    Lemma~\ref{Lem:PieQuery} are set up, one for $u$ denoted $\mathcal D_u$,
    and one for $v$ denoted $\mathcal D_v$. Then the following is done for each path $P = \delta(b_{i_1},b_{i_2})\in\mathcal P$. First, the algorithm in Lemma~\ref{Lem:CompressedStrips} is applied. Then if $\Box(b_{j(P)},b_{j(P)+1},u,v)$ is well-defined, its set of faces $F_{j(P)}$ is computed and stored; otherwise, $F_{j(P)} = \emptyset$. Similarly, if $\Box(b_{i_2-1},b_{i_2},u,v)$ is well-defined, its set of faces $F_{i_2-1}$ is computed and stored, and otherwise $F_{i_2-1} = \emptyset$. If $(b_{uv},b_{vu})\in\delta R$, $\mathcal D(u,v)$ computes and stores the set $F_{uv}$ of faces of $R$ contained in $\Box(b_{uv},b_{vu},u,v)$. This completes the
    description of the preprocessing for $\mathcal D(u,v)$. It is clear that
    preprocessing time is $O(r\log r)$ and that space is $O(r)$.

Now, consider a query specified by a vertex $w\in R$ and two distinct vertices
    $b_{i_1},b_{i_2}\in\delta R$ such that $P(b_{i_1},u)$ and $P(b_{i_2},v)$ are pairwise vertex-disjoint. First, $\mathcal D(u,v)$ identifies a boundary vertex
    $b_i$ such that $\Box(b_{i_1},b_{i_2},u,v) =
    \Delta(u,b_{i_1},b_i)\cup\Box(b_i,b_{i+1})\cup\Delta(v,b_{i+1},b_{i_2})$; this is
    possible by Lemma~\ref{Lem:SepDecomposition}. Then $\mathcal D_u$ and
    $\mathcal D_v$ are queried to determine if
    $w\in\Delta(u,b_{i_1},b_i)\cup\Delta(v,b_{i+1},b_{i_2})$; if this is the case then
    $w\in\Box(b_{i_1},b_{i_2},u,v)$ and $\mathcal D(u,v)$ answers ``yes''. Otherwise,
    $\mathcal D(u,v)$ identifies an arbitrary face $f\ne\delta R$ of $R$ incident to $w$. At this point, the only way that $w$ can belong to $\Box(b_{i_1},b_{i_2},u,v)$ is if $w$ belongs to the interior of $\Box(b_i,b_{i+1},u,v)$ which happens iff $f$ is contained in $\Box(b_i,b_{i+1},u,v)$. If $(b_i,b_{i+1}) = (b_{uv},b_{vu})$, $\mathcal D(u,v)$ checks if $f\in F_{uv}$ and if so outputs ``yes''. Otherwise, there exists a path $P\in\mathcal P$ containing $(b_i,b_{i+1})$ and $\mathcal D(u,v)$ identifies this path. It follows from Lemma~\ref{Lem:CompressedStrips} and from the definition of $\mathcal P$ that $f$ is contained in $\Box(b_i,b_{i+1},u,v)$ iff at least one of the following conditions hold:
\begin{enumerate}
\item $j_f(P)$ and $j_f'(P)$ are well-defined and $j_f'(P) < i \le j_f(P)$.
\item $f\in F_{i_2-1}$, $j_f'(P)$ is well-defined, and $i > j_f'(P)$,
\item $f\in F_{j(P)}$, $j_f(P)$ is well-defined, and $i\le j_f(P)$,
\item $f\in F_{j(P)}$ and $j_f(P)$ is undefined,
\end{enumerate}
Data structure $\mathcal D(u,v)$ checks if any one these conditions hold and if
so outputs ``yes''; otherwise it outputs ``no''. Two of the cases are
illustrated in Figure~\ref{fig:path_cases}.

\begin{figure}[htbp]
  \centering
  \includegraphics[width=.5\textwidth]{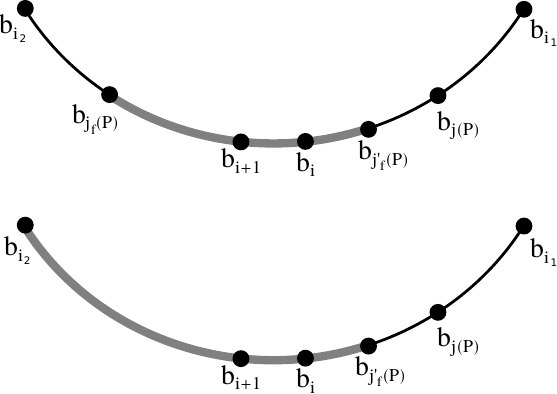}
  \caption{Illustration of how to determine if the face $f$ belongs to
    $\Box(b_i,b_{i+1},u,v)$. The illustration includes cases 1 and 4 in 
    the proof of Lemma~\ref{Lem:RegionDS}.
    The large gray subpath indicates the part, where $f$ is contained in
    each box defined by consecutive boundary nodes.}
  \label{fig:path_cases}
\end{figure}

It remains to show that $\mathcal D(u,v)$ has query time $O(1)$. By Lemma~\ref{Lem:SepDecomposition}, identifying $b_k$ can be done in $O(1)$ time. Querying $\mathcal D_u$ and $\mathcal D_v$ takes $O(1)$ time by Lemma~\ref{Lem:PieQuery}. Checking whether $f\in F_{uv}$ can clearly be done in $O(1)$ time since this set of faces is stored explicitly. With our representation of paths by the indices of their endpoints, identifying $P$ takes $O(1)$ time. Finally, since sets $F_{j(P)}$ and $F_{j(P)}$ are explicitly stored, the four conditions above can be checked in $O(1)$ time.
\end{proof}

\section{The Distance Oracle}
In this section we give a detailed presentation of both our algorithm
for answering distance queries and our distance oracle data structure.

Combining 
Lemmas~\ref{lem:spaceboundDS},~\ref{lem:querytime}
and~\ref{lem:correctness} with $r = n^{2/3}$ 
directly implies Theorem~\ref{Thm:main}.


\subsection{The Data Structure}
We present the algorithm for building our data structure.
\medskip

\textsc{Preprocessing} $G$
\begin{enumerate}
    \itemsep0pt
\item Compute an $r$-division $\calR$ of $G$. Let $\delta$ be the 
  set of all boundary vertices.
\item Store for each internal vertex the region to which it belongs.
\item\label{step:allboundaryvert}
    Compute and store the distances from each vertex to each 
  boundary vertex.
\item\label{step:allinternalpairs} For each region $R \in \calR$, 
  compute and store the distances between any pair
  of internal vertices of $R$.
\item\label{step:recursivedecomp} For each region $R$, for each
  vertex $u \notin R$, for each hole 
  $H$, compute $\text{Vor}_{H}(R,u)$ and store  a
  separator decomposition as described in 
  Section~\ref{sec:recursivedecomp}.
\item\label{step:querysep} 
  For each region $R$, for each edge $(x,y) \in R$, for each hole 
  $H$, 
  compute and store the data structure described in 
  Section~\ref{sec:PreprocRegion}.
\end{enumerate}

\begin{lemma}
  \label{lem:spaceboundDS}
  The total size of the data structure computed by \textsc{Preprocessing}
  is $O(n^2/\sqrt{r} + n \cdot r)$.
\end{lemma}
\begin{proof}
  Recall that by definition of the $r$-division, there are $O(n/\sqrt{r})$ 
  boundary vertices and $O(n/r)$ regions. 
  Thus, the number of distances 
  stored at step~\ref{step:allboundaryvert} of the 
  algorithm is at most $O(n^2/\sqrt{r})$.

  For a given region, storing the pairwise distances between all 
  its internal vertices takes $O(r^2)$ space.
  Since there are $O(n/r)$ regions in total, 
  Step~\ref{step:allinternalpairs} 
  takes memory
  $O(n \cdot r)$.

  We now bound the space taken by Step~\ref{step:recursivedecomp}.
  There are $n/r$ choices for $R$ and $n$ choices for $u$. 
  By Lemma~\ref{Lem:sepspace}, each decomposition can be stored using
  $O(\sqrt{r})$ space. Thus, this step takes 
  $O(n^2/\sqrt{r})$ total space.

  We finally bound the space taken by Step~\ref{step:querysep}.
  There $n/r$ choices for $R$ and $r$ choices for an edge $(x,y)$.
  By Lemma~\ref{Lem:RegionDS}, for a given edge $(x,y)$,
  the data structure takes $O(r)$ space.
  Hence, the total space taken by this step is $O(n \cdot r)$ 
  and the lemma follows.
\end{proof}



\begin{theorem}
  \label{thm:preprocesstime}
  The execution of \textsc{Preprocessing} takes 
  $O(n^2)$
  time  and $O(n \cdot r + n^2/\sqrt{r} )$ space.
\end{theorem}
\begin{proof}
  We analyze the procedure step by step. Computing an $r$-division 
  with $O(1)$ holes 
  can be done in linear time and space using 
  the algorithm of Klein et al., see~\cite{KleinMS13}.
  
  We now analyze Step~\ref{step:allboundaryvert}. 
  There are $O(n/\sqrt{r})$ boundary vertices. Computing single-source
  shortest paths can be done in linear time using the algorithm
  of Henzinger et al.~\cite{Henzinger97}. Hence, 
  Step~\ref{step:allboundaryvert}
  takes at most $O(n^2/\sqrt{r})$ time and space.
  
  Step~\ref{step:allinternalpairs} takes 
  $O(n \cdot r)$ time and space
  using the following algorithm. The algorithm proceed region by region
  and hole by hole.
  For a given region and hole, the algorithm adds an edge 
  between each pair
  of boundary vertices that are on the hole 
  of length equal to the distance between these
  vertices in the whole graph. Note that this is already in memory 
  and was computed at Step~\ref{step:allboundaryvert}.
  Now, for each vertex of the region, the algorithm runs a 
  shortest path algorithm. Since there are $O(\sqrt{r})$ boundary
  vertices, the number of edges added is $O(r)$. 
  Thus, the algorithm is run on a graph that has at most
  $O(r)$ edges and vertices. 
  The algorithm spends at most $O(r)$
  time per vertex of the region.
  Since there are $O(n/r)$ regions and $O(r)$ vertices
  per region, both the running time and the space
  are $O(n \cdot r)$.

  Step~\ref{step:recursivedecomp} takes $O(n^2)$ using the following 
  algorithm.
  The algorithm proceeds vertex by vertex, region by region, hole
  by hole.
  For a given vertex $u$, a given region $R$, and a given hole $H$
  the algorithm computes $\text{Vor}_{H}(R,u)$. This can be done
  by adding a ``dummy'' vertex reprensenting $u$ and connecting it to 
  each boundary vertex $x$ of the hole by an edge of length $\dist(u,x)$.
  Thus, this takes time $O(r)$ using the single-source shortest path 
  algorithm of Henzinger et al~\cite{Henzinger97}. Furthermore, 
  by Lemma~\ref{Lem:sepspace} and Corollary~\ref{Cor:Separatortime}, 
  computing the separator decomposition of 
  Section~\ref{sec:recursivedecomp} given $\text{Vor}_{H}(R,u)$ 
  takes $\tilde{O}(\sqrt{r})$ time and $O(\sqrt r)$ space.
  Thus, over all vertices, regions and holes, this step takes
  $O(n^2)$ time and $O(n^2/\sqrt{r})$ memory.

  Finally, we show that
  Step~\ref{step:querysep} takes $\tilde{O}(n \cdot r)$ time and 
  $O(n \cdot r)$ space. The algorithm proceeds region by region, 
  hole by hole, and edge by edge. By Lemma~\ref{Lem:RegionDS}, 
  for a given edge of the region, 
  computing the
  data structure of Section~\ref{sec:PreprocRegion} takes 
  $O(r \cdot \log r)$
  time and space. Since the total number of region is $O(n/r)$ 
  and the total number of edges per region is $O(r)$, the proof
  is complete.
\end{proof}

\begin{corollary}\label{cor:cabellotime}
  There exists a distance oracle with total space $O(n^{11/6})$
  and expected preprocessing time $O(n^{11/6})$.
\end{corollary}
\begin{proof}
  We apply Procedure \textsc{Preprocessing} with $r= n^{1/3}$.
  By Lemma~\ref{lem:spaceboundDS}, the total size of the 
  data structure output is $O(n^{11/6})$.

  We now analyse the total preprocessing time.
  For Steps~\ref{step:allboundaryvert},~\ref{step:allinternalpairs},
  and~\ref{step:querysep}, we mimicate the analysis of the proof
  of Theorem~\ref{thm:preprocesstime} and obtain 
  a total preprocessing time of $O(n \cdot r + n^2/\sqrt{r})$.

  We now explain how to speed-up Step~\ref{step:recursivedecomp}.
  We show that Step~\ref{step:recursivedecomp} can be done in time
  $O(n \cdot r^{5/2})$ using Cabello's data structure~\cite{Cabello17} 
  for computing weighted Voronoi diagrams of a given region. 
  
  More formally, Cabello introduces a data structure that allows to
  compute weighted Voronoi diagrams of a given region in expected 
  time $\tilde{O}(\sqrt{r})$. This data structure
  has preprocessing time $O(r^{7/2})$.
  Hence the total preprocessing time for computing the 
  data structure for all the regions is $\tilde{O}(n \cdot r^{5/2})$.

  Then, for each vertex $u$, each region $R$, each hole $H$, 
  the algorithm
  \begin{enumerate}
  \item uses the data structure to compute the weighted Voronoi
    diagram $\text{Vor}_{H}(R,u)$ in expected time
    $\tilde{O}(\sqrt{r})$ and 
  \item computes the separator decomposition of
    Section~\ref{sec:recursivedecomp} in time
    $\tilde{O}(\sqrt{r})$ (by Lemma~\ref{Lem:sepspace} 
    and Corollary~\ref{Cor:Separatortime}).
  \end{enumerate}
  This results in an expected preprocessing time of $O(n \cdot r^{5/2} + 
  n^2/\sqrt{r})$. Choosing $r= n^{1/3}$ yields a bound of
  $O(n^{11/6})$.
\end{proof}


\subsection{Algorithm for Distance Queries}\label{sec:query}
This section is devoted to the presentation of our 
algorithms for answering distance queries between pairs of vertices.

We show that any distance query between two vertices $u$, $v$ can be
performed in $O(\log r)$ time. In the following, let $u,v$ be two 
vertices of the graph. The algorithm is the following.\\
\smallskip

\textsc{Distance Query} $u,v$ 
\begin{enumerate}
\item If $u,v$ belong to the same region or if either $u$ or $v$ is 
  a boundary vertex, the query can be answered in $O(1)$ time since 
  the distances between vertices of the same region and between boundary
  vertices and the other vertices of the graph are stored explicitly.
\item If $u$ and $v$ are internal vertices that 
  belong to two different regions we proceed as 
  follows. Let $R$ be the region containing $v$ and $\delta R$ be the 
  set of boundary vertices of region $R$. The boundary vertices are
  partitioned into holes 
  $\calH = \{H_0, \ldots, H_k\}$, 
  such that $\bigcup_{H \in \calH} H  = \delta R$.
  For each $H \in \calH$, 
  we apply the following procedure.
  Let $\calV$ be the weighted Voronoi 
  diagram where the sites are the vertices of $H$ and the 
  weight of $x \in H$ is the distance from $u$ to $x$.

  We now aim at determining to which cell of $\calV$, $v$ 
  belongs. We use the binary search procedure of 
  Lemma~\ref{lem:mainseparator} on the
  decomposition of $R$ induced by the separators of the weighted
  Voronoi diagram.
  More precisely, we use the algorithm described in 
  Section~\ref{sec:recursivedecomp}, Lemma~\ref{lem:mainseparator}, and
  the query algorithm described in Section~\ref{sec:PreprocRegion}, 
  Lemma~\ref{Lem:RegionDS} to identify a set of 
  at most six Voronoi cells so that one of them contains 
  $v$. This induces a set of at most six boundary vertices
  $X = \{x_0,\ldots,x_k\}$ that represent the centers of the cells.

  Finally, we have the distances from both $u$ and $v$ to 
  all the boundary vertices in $X$.  
  Let $v(H) = \min_{x \in X} \dist(u,x) + \dist(x,v)$.
  The algorithm returns $\min_{H \in \calH} v(H)$.
\end{enumerate}

\begin{lemma}[Running time]
  \label{lem:querytime}
  The \textsc{Distance Query} takes $O(\log r)$ time.
\end{lemma}
\begin{proof}
  Consider a distance query from a vertex $u$ to a vertex $v$ and assume
  that those vertices are internal vertices of two different regions
  as otherwise the query takes $O(1)$ time.
  Observe that we can determine in $O(1)$ time to which region $v$ 
  belongs. Fix a hole $H$.
  Let $\calV$ be the weighted Voronoi 
  diagram where the sites are the vertices of $H$ and the 
  weight of $x \in \delta R$ is the distance from $u$ to $x$.
  We consider the decomposition of the region of $v$ of 
  $\Vor{R}{u}$.


  Lemma~\ref{Lem:RegionDS} shows that the query time for the data 
  structure defined in Section~\ref{sec:PreprocRegion} is $t = O(1)$.
  Applying Lemma~\ref{lem:mainseparator} with $t=O(1)$ implies 
  that the total time to determine
  in which Voronoi cell $v$ belongs is at most $O(\log r)$.
  
  Finally, computing $\min_{x \in X} \dist(u,x) + \dist(x,v)$ takes $O(1)$ 
  time. By definition of the $r$-division 
  there are $O(1)$ holes.

  Therefore, 
  we conclude that the running time of the \textsc{Distance Query} 
  algorithm is $O(\log r)$.
\end{proof}

We now prove that the algorithm indeed returns the correct distance
between $u$ and $v$.
\begin{lemma}[Correctness]
  \label{lem:correctness}
  The \textsc{Distance Query} on input $u,v$ returns the length 
  of the shortest path between vertices $u$ and $v$ in the graph.
\end{lemma}
\begin{proof}
  We remark that the distance from any vertex to a boundary vertex is 
  stored explicitly and thus correct.
  Hence, we consider the case where $u$ and $v$ are internal 
  vertices of different regions. 
  Let $P$ be the shortest path from $u$ to $v$ in $G$.
  Let $x \in P$ be the last boundary vertex of $R$ on the path from $u$ 
  to $v$ and let $H_x$ be the hole containing $x$. 
  Let $\calV$ be
  the weighted Voronoi 
  diagram
  where the sites are the boundary vertices of $H_x$ and the 
  weight of $y \in H_x$ is the distance from $u$ to $y$.

  We need to argue that the data structure of 
  Section~\ref{sec:PreprocRegion}, Lemma~\ref{Lem:RegionDS} satisfies
  the conditions of Lemma~\ref{lem:mainseparator}.
  Observe that the separators defined in 
  Section~\ref{sec:recursivedecomp} consist of two shortest
  paths $P_R(b_1,x)$ and $P_R(b_2,y)$ where $b_1,b_2 \in \delta R$ and
    $(y,z)$
  is an edge of $R$. Hence, the set of vertices of the subgraph 
  $\Box(b_1,b_2,y,z)$ correspond to the set of 
  all the vertices that are one of the two sides of the separator.
  Thus, by Lemma~\ref{Lem:RegionDS} 
  the data structure described in Section~\ref{sec:PreprocRegion}
  satisfies the condition of Lemma~\ref{lem:mainseparator},
  with query time $t = O(1)$.

  We argue that $v$ belongs to the Voronoi cell of $x$. 
  Assume towards contradiction that $v$ is in the Voronoi cell
  of $y \neq x$, we would have 
  $\dist(y, v) + w(y) \le \dist(x,v) + w(x)$,
  where $w$ is the weight function associated with the Voronoi diagram.
  Thus, this implies that
  $\dist(y, v) + \dist(y,u) \le \dist(x,v) + \dist(x,u)$.
  Therefore, there exists a shortest path from $u$ to $v$ that 
  goes through $y$.
  Now observe that if $v$ belongs to the Voronoi cell of $y$, the 
  shortest path from $v$ to $y$ does not go through $x$.
  Hence, assuming unique shortest paths between pairs of vertices, 
  we conclude that the last boundary vertex on the path
  from $u$ to $v$ is $y$ and not $x$, a contradiction. Thus,
  $v$ belongs to the Voronoi cell of $x$.

  Combining with Lemma~\ref{Lem:RegionDS}, it follows that
  the Voronoi cell of $x$ is in the set of Voronoi cells $X$ 
  obtained at the end of the recursive procedure.
  Observe that for any $x' \in X$ there exists a path 
  (possibly with repetition of vertices) of length 
  $\dist(x',u) + \dist(x',v)$.
  Therefore, since we assume unique shortest paths between pairs of 
  vertices, we conclude that
  $\dist(u,v) = 
  \dist(x,u) + \dist(x,v) = \min_{x' \in X} \dist(x',v) +
  \dist(x',u) = v(H_x) = \min_{H} v(H)$.
\end{proof}


  

\section{Trade-off}\label{sec:tradeoff}
We now prove \Cref{thm:tradeoff}. We first consider the case with $P = O(n^2)$
and then extend it to the case with efficient preprocessing time.

Let $r_1\le r_2$ be positive integers to be defined later. The case $r_1\le
r_2$ will correspond exactly to $S\ge n\sqrt{n}$. The data structure works as
follows.
\begin{enumerate}
    \item We compute an $r_1$-division and an $r_2$-division of $G$ named
        $\mathcal{R}_1$ and $\mathcal{R}_2$ respectively. Let $\delta_1$
        respectively $\delta_2$ denote all the boundary vertices of
        $\mathcal{R}_1$ resp. $\mathcal{R}_2$.
    \item For each region $R$ of $\mathcal{R}_1$ and $\mathcal{R}_2$, compute
        and store the pairwise distances of the nodes of $R$.
    \item For each $u\in \delta_1$ and $v\in \delta_2$, compute and store the
        distance between $u$ and $v$ in $G$.
    \item For each $u\in \delta_1$, region $R\in \mathcal{R}_2$ and hole $H\in
        R$, compute $\text{Vor}_{H}(R,u)$ and store a separator decomposition
        as described in Section~\ref{sec:recursivedecomp}.
    \item For each Region $R\in\mathcal{R}_2$, for each edge $(x,y)\in R$, for
        each hole $H$, compute and store the data structure described in
        Section~\ref{sec:PreprocRegion}.
\end{enumerate}

We start by bounding the space.
\begin{lemma}\label{lem:tradeoff_size}
    The total size of the data structure described above is
    \[
        O\!\left(nr_2 + \frac{n^2}{\sqrt{r_1r_2}}\right)\ .
    \]
\end{lemma}
\begin{proof}
    Consider the steps above. By definition of $\mathcal{R}_1$ and
    $\mathcal{R}_2$ step 2 uses $O(nr_2)$ space (since we assumed $r_1\le r_2$)
    and step 4 uses $\frac{n}{\sqrt{r_1}} \cdot \frac{n}{\sqrt{r_2}}$ space. By
    Lemma~\ref{Lem:sepspace} step 5 takes $O(n/\sqrt{r_2})$ for each node of
    $\delta_1$ giving $O(n^2/\sqrt{r_1r_2})$ in total. By
    Lemma~\ref{Lem:RegionDS} the total space of step 6 is $O(nr_2)$.
\end{proof}

Now consider a query pair $u,v$. If $u$ and $v$ belong to the same region in
$\mathcal{R}_1$ or $\mathcal{R}_2$ we return the stored distance.
Otherwise we iterate over each boundary node $w$ in the region of $u$ in
$\mathcal{R}_1$. For each such boundary node we compute the distance to $v$
using the data structures of steps 5 and 6 above similar to the
query algorithm from Section~\ref{sec:query}. This is possible since we have
stored the distances between all the needed boundary nodes in step 4. The
minimum distance returned over all such $w$ is the answer to the query.

From the description above it is clear that we get a query time of
$O(\sqrt{r_1}\log(r_2))$. The correctness follows immediately from the
discussion in the proof of Theorem~\ref{Thm:main}. What is left now is to
balance the space to obtain Theorem~\ref{thm:tradeoff}.
The expression of Lemma~\ref{lem:tradeoff_size} is balanced when
\[
    r_2 = \frac{n^{2/3}}{r_1^{1/3}}\ .
\]
Now, since we assumed that $S \ge n\sqrt{n}$ we can focus on the case when
$r_2\ge \sqrt{n}$ and thus we get $r_1\le r_2$ as we required. Plugging into
the definition of $Q$ we get exactly
\[
    Q = O(\sqrt{r_1}\log n) = \frac{n^{5/2}}{S^{3/2}}\log n\ ,
\]
which gives us Theorem~\ref{thm:tradeoff}.

For pre-processing time, we consider two cases similar to
\Cref{thm:preprocesstime} and \Cref{cor:cabellotime}. It follows directly from
the discussion above and \Cref{thm:preprocesstime} that the preprocessing can
be performed in $O(n^2)$. We may, however, also consider pre-processing time as
a parameter similar to space and query time. This gives a 3-way trade-off. In
\Cref{cor:cabellotime} we showed how to lower pre-processing time by increasing
the space. Here we discuss the case of lowering pre-processing time further by
increasing the query time.
It follows from the discussion above and \Cref{cor:cabellotime} that we can
perform pre-processing of the above structure in $O(nr_1 + nr_2^{5/2} +
n^2/\sqrt{r_1r_2})$ time and get the same space bound. If we assume that $r_1
\le r_2^{5/2}$ we get a data structure with query time $Q = O(\sqrt{r_1}\log
n)$ and space and pre-processing time $S = O(nr_2^{5/2} + n^2/\sqrt{r_1r_2})$.
Up to logarithimic and constant factors,
this gives us $Q = n^{11/5}/S^{6/5}$. For any $S \ge n^{16/11}$ this satisfies
the requirement that $r_1\le r_2^{5/2}$.
As an example, we get a data structure with space and pre-processing
time $O(n^{16/11})$ and a query time of $O(n^{5/11}\log n)$.

\newpage
\bibliographystyle{plain}
\bibliography{planar}

\begin{thebibliography}{10}

\bibitem{TSPlib}
The {TSPLIB}.
\newblock \url{http://comopt.ifi.uni-heidelberg.de/software/TSPLIB95/}.
\newblock Accessed: 2010-09-30.

\bibitem{AbboudD16}
Amir Abboud and S{\o}ren Dahlgaard.
\newblock Popular conjectures as a barrier for dynamic planar graph algorithms.
\newblock In {\em Proc. 57th IEEE Symposium on Foundations of Computer Science
  (FOCS)}, pages 477--486, 2016.

\bibitem{ArikatiCCDSZ96}
Srinivasa~Rao Arikati, Danny~Z. Chen, L.~Paul Chew, Gautam Das, Michiel H.~M.
  Smid, and Christos~D. Zaroliagis.
\newblock Planar spanners and approximate shortest path queries among obstacles
  in the plane.
\newblock In {\em Proc. 4th European Symposium on Algorithms (ESA)}, pages
  514--528, 1996.

\bibitem{Cabello12}
Sergio Cabello.
\newblock Many distances in planar graphs.
\newblock {\em Algorithmica}, 62(1-2):361--381, 2012.
\newblock See also SODA'06.

\bibitem{Cabello17}
Sergio Cabello.
\newblock Subquadratic algorithms for the diameter and the sum of pairwise
  distances in planar graphs.
\newblock In {\em Proc. 28th ACM/SIAM Symposium on Discrete Algorithms (SODA)},
  pages 2143--2152, 2017.

\bibitem{Chechik14}
Shiri Chechik.
\newblock Approximate distance oracles with constant query time.
\newblock In {\em Proc. 46th ACM Symposium on Theory of Computing (STOC)},
  pages 654--663, 2014.

\bibitem{Chechik15}
Shiri Chechik.
\newblock Approximate distance oracles with improved bounds.
\newblock In {\em Proc. 47th ACM Symposium on Theory of Computing (STOC)},
  pages 1--10, 2015.

\bibitem{ChenX00}
Danny~Z. Chen and Jinhui Xu.
\newblock Shortest path queries in planar graphs.
\newblock In {\em Proc. 22nd ACM Symposium on Theory of Computing (STOC)},
  pages 469--478, 2000.

\bibitem{Cohen-AddadDW17}
Vincent Cohen{-}Addad, S{\o}ren Dahlgaard, and Christian Wulff{-}Nilsen.
\newblock Fast and compact exact distance oracle for planar graphs.
\newblock {\em CoRR}, abs/1702.03259, 2017.

\bibitem{Djidjev96}
Hristo Djidjev.
\newblock On-line algorithms for shortest path problems on planar digraphs.
\newblock In {\em Graph-Theoretic Concepts in Computer Science, 22nd
  International Workshop, {WG} '96, Cadenabbia (Como), Italy, June 12-14, 1996,
  Proceedings}, pages 151--165, 1996.

\bibitem{erdHos1964extremal}
Paul Erd{\H{o}}s.
\newblock Extremal problems in graph theory.
\newblock In {\em IN “THEORY OF GRAPHS AND ITS APPLICATIONS,” PROC. SYMPOS.
  SMOLENICE}. Citeseer, 1964.

\bibitem{FakR06}
Jittat Fakcharoenphol and Satish Rao.
\newblock Planar graphs, negative weight edges, shortest paths, and near linear
  time.
\newblock {\em Journal of Computer and System Sciences}, 72(5):868--889, 2006.
\newblock See also FOCS'01.

\bibitem{Frederickson87}
Greg~N. Frederickson.
\newblock Fast algorithms for shortest paths in planar graphs, with
  applications.
\newblock {\em SIAM Journal on Computing}, 16(6):1004--1022, 1987.

\bibitem{GawrK16}
Pawel Gawrychowski and Adam Karczmarz.
\newblock Improved bounds for shortest paths in dense distance graphs.
\newblock {\em CoRR}, abs/1602.07013, 2016.

\bibitem{GawrychowskiKMSW17}
Paweł Gawrychowski, Haim Kaplan, Shay Mozes, Micha Sharir, and Oren Weimann.
\newblock Voronoi diagrams on planar graphs, and computing the diameter in
  deterministic o~(n5/3) time, 2017.
\newblock arXiv preprint.

\bibitem{Henzinger97}
Monika~R Henzinger, Philip Klein, Satish Rao, and Sairam Subramanian.
\newblock Faster shortest-path algorithms for planar graphs.
\newblock {\em Journal of Computer and System Sciences}, 55(1):3--23, 1997.

\bibitem{ItalianoNSW11}
Giuseppe~F Italiano, Yahav Nussbaum, Piotr Sankowski, and Christian
  Wulff-Nilsen.
\newblock Improved algorithms for min cut and max flow in undirected planar
  graphs.
\newblock In {\em Proc. 43rd ACM Symposium on Theory of Computing (STOC)},
  pages 313--322, 2011.

\bibitem{KaplanMNS12}
Haim Kaplan, Shay Mozes, Yahav Nussbaum, and Micha Sharir.
\newblock Submatrix maximum queries in monge matrices and monge partial
  matrices, and their applications.
\newblock In {\em Proc. 23rd ACM/SIAM Symposium on Discrete Algorithms (SODA)},
  pages 338--355, 2012.

\bibitem{kawarabayashi2011linear}
Ken-ichi Kawarabayashi, Philip~N Klein, and Christian Sommer.
\newblock Linear-space approximate distance oracles for planar, bounded-genus
  and minor-free graphs.
\newblock In {\em International Colloquium on Automata, Languages, and
  Programming}, pages 135--146. Springer, 2011.

\bibitem{kawarabayashi2013more}
Ken-ichi Kawarabayashi, Christian Sommer, and Mikkel Thorup.
\newblock More compact oracles for approximate distances in undirected planar
  graphs.
\newblock In {\em Proceedings of the Twenty-Fourth Annual ACM-SIAM Symposium on
  Discrete Algorithms}, pages 550--563. SIAM, 2013.

\bibitem{Klein02}
Philip~N. Klein.
\newblock Preprocessing an undirected planar network to enable fast approximate
  distance queries.
\newblock In {\em Proc. 13th ACM/SIAM Symposium on Discrete Algorithms (SODA)},
  pages 820--827, 2002.

\bibitem{Klein05}
Philip~N Klein.
\newblock Multiple-source shortest paths in planar graphs.
\newblock In {\em Proc. 16th ACM/SIAM Symposium on Discrete Algorithms (SODA)},
  volume~5, pages 146--155, 2005.

\bibitem{KleinMS13}
Philip~N Klein, Shay Mozes, and Christian Sommer.
\newblock Structured recursive separator decompositions for planar graphs in
  linear time.
\newblock In {\em Proc. 45th ACM Symposium on Theory of Computing (STOC)},
  pages 505--514. ACM, 2013.

\bibitem{Klein89}
Rolf Klein.
\newblock {\em Concrete and Abstract Voronoi Diagrams}, volume 400 of {\em
  Lecture Notes in Computer Science}.
\newblock Springer, 1989.

\bibitem{KleinLN09}
Rolf Klein, Elmar Langetepe, and Zahra Nilforoushan.
\newblock Abstract voronoi diagrams revisited.
\newblock {\em Comput. Geom.}, 42(9):885--902, 2009.

\bibitem{MozesS12}
Shay Mozes and Christian Sommer.
\newblock Exact distance oracles for planar graphs.
\newblock In {\em Proc. 23rd ACM/SIAM Symposium on Discrete Algorithms (SODA)},
  pages 209--222, 2012.

\bibitem{Nussbaum11}
Yahav Nussbaum.
\newblock Improved distance queries in planar graphs.
\newblock In {\em Proc. 12th Workshop on Algorithms and Data Structures
  (WADS)}, pages 642--653, 2011.

\bibitem{Thorup04}
Mikkel Thorup.
\newblock Compact oracles for reachability and approximate distances in planar
  digraphs.
\newblock {\em Journal of the ACM}, 51(6):993--1024, 2004.
\newblock See also FOCS'01.

\bibitem{ThorupZ05}
Mikkel Thorup and Uri Zwick.
\newblock Approximate distance oracles.
\newblock {\em Journal of the ACM}, 52(1):1--24, 2005.
\newblock See also STOC'01.

\bibitem{WNthesis}
Christian Wulff-Nilsen.
\newblock {\em Algorithms for Planar Graphs and Graphs in Metric Spaces}.
\newblock PhD thesis, University of Copenhagen, 2010.

\bibitem{wulff2012approximate}
Christian Wulff-Nilsen.
\newblock Approximate distance oracles with improved preprocessing time.
\newblock In {\em Proceedings of the twenty-third annual ACM-SIAM symposium on
  Discrete Algorithms}, pages 202--208. Society for Industrial and Applied
  Mathematics, 2012.

\bibitem{WN16}
Christian Wulff{-}Nilsen.
\newblock Approximate distance oracles for planar graphs with improved query
  time-space tradeoff.
\newblock In {\em Proc. 27th ACM/SIAM Symposium on Discrete Algorithms (SODA)},
  pages 351--362, 2016.

\end{thebibliography}

\newpage
\appendix

\end{document}